\documentclass[aps,pra,showpacs,showkeys,linenumbers]{revtex4}
\usepackage{amssymb,amsfonts,amsmath,enumerate,amsthm}
\usepackage{graphicx}
\usepackage{verbatim}
\usepackage{comment}
\usepackage{color}

\newcommand{\Tr}{\mathrm{Tr}\,}
\newtheorem*{proposition}{Proposition}

\newcommand{\bra}[1]{\langle #1|}
\newcommand{\ket}[1]{|#1\rangle}
\newcommand{\braket}[2]{\langle #1|#2\rangle}
\newcommand{\ketbra}[2]{|#1 \rangle \langle #2|}

\begin{document}

\title{A geometric formulation of the uncertainty principle}

\author{G.M. Bosyk $^\dagger$, T.M. Os\'an $^\ddagger$, P.W. Lamberti $^\ddagger$ and M. Portesi $^\dagger$}

\affiliation{$^\dagger$ Instituto de F\'\i{}sica La Plata (IFLP), CONICET, and Departamento de F\'{\i}sica, Facultad de Ciencias Exactas, Universidad Nacional de La Plata, 115 y 49, C.C.~67, 1900 La Plata, Argentina}

\affiliation{$^\ddagger$ Facultad de Matem\'atica, Astronom\'{\i}a y F\'{\i}sica, Universidad Nacional de C\'{o}rdoba and CONICET, Av. Medina Allende s/n, Ciudad Universitaria, X5000HUA C\'ordoba, Argentina}

\begin{abstract}

A geometric approach to formulate the uncertainty principle between quantum observables acting on an $N$-dimensional Hilbert space is proposed. We consider the fidelity between a density operator associated with a quantum system and a projector associated with an observable, and interpret it as the probability of obtaining the outcome corresponding to that projector.
We make use of fidelity-based metrics such as angle, Bures and root-infidelity ones, to propose a measure of uncertainty. The triangle inequality allows us to derive a family of uncertainty relations. In the case of the angle metric, we re-obtain the Landau--Pollak inequality for pure states and show, in a natural way, how to extend it to the case of mixed states in arbitrary dimension. In addition, we derive and compare novel uncertainty relations when using other known fidelity-based metrics.


\end{abstract}

\date{\today}

\maketitle

\section{Introduction}
\label{sec:intro}

 The uncertainty principle is one of the major features of quantum mechanics, establishing a limitation on the predictability of incompatible observables. {\em Uncertainty relations} (URs) constitute the mathematical formulation of this principle. Variance-based URs, such as those of Heisenberg, Robertson and Schr\"odinger~\cite{Hei27Rob29Sch30}, are the most popular ones; they exhibit a state-dependent lower bound for the product of the variances of a pair of non-commuting observables.
 Even though this kind of URs allow us to form mental pictures useful to get insights about quantum theory and also provide means for important quantitative predictions, they do not always capture the essence of the principle, as has been pointed out in Refs.~\cite{Deutsch1983,UffinkThesis}. Accordingly, a variety of alternative formulations have been proposed such as those using higher-order moments~\cite{highmoments}, or information-theoretic measures of ignorance that give entropic uncertainty inequalities~\cite{EURs}. The relevance of the study of URs relies to some extent on the fact that they are useful in several applications of quantum information as well as entanglement detection and quantum cryptography, among many others~\cite{applications}.

 The geometric approach to quantum mechanics plays a fundamental role not only in foundational issues but also in applications of quantum information processing~\cite{GeometryBook}. In order to contribute within this approach, we provide here a geometric formulation of the uncertainty principle. Our proposal is inspired by the Landau--Pollak inequality, which has been introduced in time--frequency analysis~\cite{Landau1961} and later adapted to quantum mechanics by Maassen and Uffink~\cite{Maassen1988}.

 This work is organized as follows: in Sec.~\ref{sec:framework} we establish the theoretical framework, recalling concepts like purification of mixed quantum states and the notion of fidelity along with some of its expressions and properties, summarizing some known fidelity-based metrics, as well as the Landau--Pollak inequality for pure states. All these concepts are used in Sec.~\ref{sec:derivation}, which contains our major contributions: a proof of Landau--Pollak inequality for mixed quantum states, and a geometric derivation of URs. Finally, some conclusions are drawn in Sec.~\ref{sec:conclu}.

\section{Theoretical framework}
\label{sec:framework}

\subsection{Purification of mixed quantum states and fidelity}
\label{ssec:purif&fidelity}

Purification is a procedure which allows us to associate {\em mixed} quantum states with {\em pure} quantum states. This concept emerges from the fact that, given a density matrix $\rho$ acting on a Hilbert space $\mathcal{H}$ and associated with a {\em mixed} quantum state, it is always possible to find a {\em pure} quantum state $\ket{\Psi}$ belonging to an {\em extended} Hilbert space $\mathcal{H} \otimes \mathcal{H}_{\rm aux}$ so that
\begin{equation}
\Tr_{\rm aux} (\ketbra{\Psi}{\Psi}) = \rho\nonumber
\end{equation}
where $\Tr_{\rm aux}$ denotes the {\em partial trace} over the degrees of freedom of the Hilbert space $ \mathcal{H}_{\rm aux}$.

In order to see how to purify a mixed quantum state by means of a mathematical procedure see, for example, Refs.~\cite{Jozsa,Nielsen&ChuangBook}.  In this case, it can be seen that purification is not unique. It is also important to note that purification is not an abstract concept beyond physical reality. For example, in Ref.~\cite{Bassi}, a discussion regarding physical purifications of a given ensemble can be found.

The {\em fidelity} $F(\rho,\sigma)$ is a well known measure of {\em similarity} between two quantum states represented by density matrices $\rho$ and $\sigma$. This quantity is defined as~\cite{Uhlmann,Jozsa}:
\begin{equation}
\label{fidelity}
F(\rho,\sigma)=\left(\Tr\sqrt{\sqrt{\rho} \sigma \sqrt{\rho}} \right)^2 .
\end{equation}
An equivalent definition of $F$ can be provided in terms of purifications of the states $\rho$ and $\sigma$~\cite{Jozsa}. Thus, we can write
\begin{equation}
\label{UJFdef1}
F(\rho, \sigma) = \max_{\ket{\psi},\ket{\varphi}}\left|\braket{\psi}{\varphi}\right|^2
\end{equation}
where maximization is performed over {\em all} purifications $\ket{\psi}$ of $\rho$ and $\ket{\varphi}$ of $\sigma$. 
Another equivalent definition of $F$ in terms of purifications of the states $\rho$ and $\sigma$ is~\cite{Jozsa}
\begin{equation}
\label{UJFdef2}
F(\rho, \sigma) = \max_{\ket{\varphi}}\left|\braket{\psi}{\varphi}\right|^2
\end{equation}
where $\ket{\psi}$ is any {\em fixed} purification of $\rho$ and maximization is performed over {\em all} purifications $\ket{\varphi}$ of $\sigma$.

It can be shown that fidelity satisfies the following properties~\cite{Nielsen&ChuangBook,Jozsa}:
\begin{enumerate}
\item \label{FP1} {\em Normalization:}
\begin{equation}
0 \leq F(\rho, \sigma) \leq 1
\end{equation}
\item \label{FP2} {\em Identity of indiscernibles:}
\begin{equation}
F(\rho, \sigma) = 1 \:\:{\rm if\:\:and\:\:only\:\:if}\:\:\rho = \sigma
\end{equation}
\item \label{FP3} {\em Symmetry:}
\begin{equation}
F(\rho, \sigma) = F(\sigma, \rho)
\end{equation}
\item \label{FP4} If $\rho = \ketbra{\psi}{\psi}$ and $\sigma = \ketbra{\varphi}{\varphi}$ represent pure states, fidelity reduces to
\begin{equation}
\label{fidelity_pure}
F(\ketbra{\psi}{\psi},\ketbra{\varphi}{\varphi}) = \left|\braket{\psi}{\varphi} \right|^2
\end{equation}
\item \label{FP5} If $\rho = \ketbra{\psi}{\psi}$ represents a pure state then $F(\ketbra{\psi}{\psi}, \sigma) = \bra{\psi}\sigma\ket{\psi} = \Tr(\rho \,\sigma)$
\end{enumerate}

For a more complete list of fidelity properties see, for example, Refs.~\cite{Nielsen&ChuangBook,Jozsa,OLPRA2013}.

\subsection{Fidelity-based metrics}
\label{ssec:f-metrics}

The fidelity $F$ serves as a generalized measure of similarity between two quantum states but is important to notice that $F$ is not a metric. Nevertheless, it can be used to define metrics on the space of quantum states. We recall that a (true) metric is a function $d(\rho,\sigma)$ that verifies the following properties:
\begin{enumerate}[i.]
\item {\em non-negativity}: $d(\rho,\sigma)\geq 0$, and $d(\rho,\sigma)=0$ if and only if $\rho=\sigma$,
\item {\em symmetry}: $d(\rho,\sigma) = d(\sigma,\rho)$,
\item {\em triangle inequality}: $d(\sigma,\rho) + d(\tau,\rho) \geq d(\sigma,\tau)$.
\end{enumerate}

In what follows, we consider some fidelity-based metrics that will be used in Sec.~\ref{sec:derivation} to derive URs.

\begin{itemize}

\item {Angle metric:}

The angle metric~\cite{Nielsen&ChuangBook,GeometryBook} can be written in terms of fidelity as
\begin{equation}
\label{eq:anglemetric}
d_{\rm A}(\rho,\sigma) = \arccos \sqrt{ F(\rho,\sigma)}.
\end{equation}
In the case of two pure states $\ket{\psi}$ and $\ket{\varphi}$ the angle metric reduces to the Wootters' metric~\cite{Wootters1981,Majtey2005}
\begin{equation}
d_{\rm W}(\ket{\varphi},\ket{\psi}) = \arccos|\braket{\varphi}{\psi}| .
\label{eq:Wmetric}
\end{equation}

\item{Bures metric:}

The Bures metric \cite{Bures1969,GeometryBook} written in terms of fidelity takes the form
\begin{equation}
\label{eq:buresmetric}
d_{\rm B} (\rho,\sigma) = \sqrt{2-2 \sqrt{F(\rho,\sigma)}}.
\end{equation}

\item{Root-infidelity metric:}

The root-infidelity metric proposed by Gilchrist \textit{et al.} in~\cite{Gilchrist2005} can be written in terms of fidelity as
\begin{equation}
\label{eq:cmetric}
 d_{\rm RI} (\rho,\sigma) = \sqrt{ 1- F(\rho,\sigma)}.
\end{equation}

\end{itemize}

\subsection{Landau--Pollak inequality for pure states}
\label{ssec:LPineqPS}

Let us consider a quantum system with states belonging to an $N$-dimensional Hilbert space $\mathcal{H}$ and two observables $\mathcal{A}$ and $\mathcal{B}$, with discrete non-degenerate spectra, acting on $\mathcal{H}$. Let $\{|a_i\rangle\}$ and $\{|b_j\rangle\}$ be the eigenbasis of $\mathcal{A}$ and $\mathcal{B}$, respectively, and $\ket{\Psi}$ be a pure state of the quantum system. Then, the Landau--Pollak inequality~(LPI) reads~\cite{UffinkThesis}
\begin{equation}
 \arccos \sqrt{P_{\mathcal{A};\Psi}} + \arccos \sqrt{P_{\mathcal{B};\Psi}} \geq \arccos c \label{eq:LP}
\end{equation}
where $\ P_{\mathcal{A};\Psi}\equiv\max_i p_i(\mathcal{A};\Psi)=\max_i |\langle a_i |\Psi \rangle|^2\,\in\,\left[\frac 1N,1\right]$ (and analogously for $\mathcal{B}$), and \ $c\equiv\max_{ij} |\langle a_i| b_j \rangle|\,\in\,\left[\frac{1}{\sqrt N},1\right]$ is the so-called \textit{overlap} between the eigenbasis of the two observables. The LPI~\eqref{eq:LP} is indeed an alternative formulation of the uncertainty principle for pure states~\cite{UffinkThesis}. Furthermore, it has been used to obtain entanglement criteria~\cite{LPent} and to improve the Maassen--Uffink entropic UR~\cite{deVicente2008}.
Now, we assert that LPI~\eqref{eq:LP} is nothing but a consequence of the triangle inequality verified by Wootters' metric~\eqref{eq:Wmetric}. Indeed, this inequality for the case of two arbitrary eigenstates $\ket{a_i}$, $\ket{b_j}$ of $\mathcal{A}$ and $\mathcal{B}$ respectively, and a given state $\ket{\Psi}$, reads
\begin{equation}
\arccos |\braket{a_i}{\Psi}| + \arccos |\braket{b_j}{\Psi}| \geq \arccos(|\braket{a_i}{b_j}|).
\end{equation}
Choosing in particular those indices $i_{max}$ and $j_{max}$ that correspond to $\max_i p_i(\mathcal{A};\Psi)$ and $\max_j p_j(\mathcal{B};\Psi)$, we have
\begin{equation}
 \arccos \sqrt{P_{\mathcal{A};\Psi}} + \arccos \sqrt{P_{\mathcal{B};\Psi}} \geq \arccos |\braket{a_{i_{max}}}{b_{j_{max}}}|.
\end{equation}
Thereby, taking into account that $|\braket{a_{i_{max}}}{b_{j_{max}}}| \leq \max_{i,j} |\braket{a_i}{b_j}| = c$ and $\arccos x$ is a decreasing function of its argument, it is straightforward to obtain LPI~\eqref{eq:LP}.

In Sec.~\ref{sec:derivation}, with the purpose of obtainig URs valid for pure as well as for mixed quantum states, we will extend the lines of the reasoning used here to derive the inequality~\eqref{eq:LP} to other metrics commonly used in quantum mechanics.

\section{Geometric derivation of uncertainty relations}
\label{sec:derivation}

\subsection{Landau--Pollak inequality for mixed states}
\label{ssec:LPineqMS}

Here we derive an extension of the LPI, Eq. \eqref{eq:LP}, for mixed quantum states acting on a Hilbert space of arbitrary dimension. Even though the metric character of the angle metric is well known (see for example Ref.~\cite{Nielsen&ChuangBook}) for the purpose of this work we find enlighting to prove it by using the concept of purification introduced in Sec.~\ref{ssec:purif&fidelity}.
Let us consider the density matrices $\rho$, $\sigma$ and $\tau$ associated with three mixed states, and let $\ket{r}$, $\ket{s}$ and $\ket{t}$ be {\em arbitrary} purifications of $\rho$, $\sigma$ and $\tau$, respectively. Then, due to the metric character of Wotters' distance $d_{\rm W}$ for pure states, the triangle inequality reads
\begin{equation}
\label{eq:LPImix1}
\arccos |\braket{s}{r}| + \arccos |\braket{t}{r}| \geq \arccos |\braket{s}{t}| .
\end{equation}
We now show that, when choosing adequately the purifications corresponding to $\sigma$ and $\tau$, each term on the left hand-side can be written in terms of the fidelity between the corresponding density matrix and $\rho$. Indeed, for any given purification $\ket{r}$ of $\rho$, we select those purifications of $\sigma$ and $\tau$ that maximize the overlaps $|\braket{s}{r}|$ and $|\braket{t}{r}|$ independently. Recalling Eq.~\eqref{UJFdef2}, this is equivalent to calculate $\sqrt{F(\sigma,\rho)}$ and $\sqrt{F(\tau,\rho)}$, respectively. Thus, we obtain
\begin{equation}
\label{eq:LPImix2}
\arccos \sqrt{F(\sigma,\rho)} + \arccos \sqrt{F(\tau,\rho)} \geq \arccos |\braket{\tilde{s}}{\tilde{t}}|
\end{equation}
where $\ket{\tilde{s}}$ and $\ket{\tilde{t}}$ denote the purifications that maximize the overlaps.
Now, due to arccosine is a decreasing function, we have
\begin{equation}
\label{eq:LPImix3}
\arccos |\braket{\tilde{s}}{\tilde{t}}| \geq \arccos \sqrt{F(\sigma,\tau)} ,
\end{equation}
which combined with~\eqref{eq:LPImix2} gives
\begin{equation}
\label{eq:LPImix4}
\arccos \sqrt{F(\sigma,\rho)} + \arccos \sqrt{F(\tau,\rho)} \geq \arccos \sqrt{F(\sigma,\tau)} .
\end{equation}
This completes the proof of the triangle inequality for the angle metric $d_\textrm{A}$.

Now, we use Eq.~\eqref{eq:LPImix4} when $\sigma=\Pi^{\mathcal{A}}_i=\ketbra{a_i}{a_i}$ \ and \ $\tau=\Pi^{\mathcal{B}}_j=\ketbra{b_j}{b_j}$. Here, the operator $\Pi^{\mathcal{A}}_i$ represents the rank-one projector associated with the $i$th outcome of $\mathcal{A}$. Thus, for a system with density matrix $\rho$ one has
\begin{equation}
\label{eq:Fprob}
F(\Pi^\mathcal{A}_i,\rho)= \Tr (\Pi^\mathcal{A}_i \rho) = p_i(\mathcal{A};\rho) .
\end{equation}
An analogous relation holds for the projector correponding to the $j$th outcome of $\mathcal{B}$.
Following the same lines of reasoning that we developed to demonstrate LPI in Sec.~\ref{ssec:LPineqPS}, we find
\begin{equation}
\label{eq:LPImix7}
\arccos \sqrt{P_{\mathcal{A};\rho}} + \arccos \sqrt{P_{\mathcal{B};\rho}} \geq \arccos c
\end{equation}
where $P_{\mathcal{A};\rho} = \max_i p_i(\mathcal{A};\rho)\,\in\,\left[\frac 1N,1\right]$ (an analogous relation holds for $\mathcal{B}$), and the overlap in terms of the projectors reads $c = \max_{i,j}\sqrt{\Tr(\Pi^{\mathcal{A}}_i \Pi^{\mathcal{B}}_j)}$. Inequality~\eqref{eq:LPImix7} is an UR and it is the natural extension of LPI~\eqref{eq:LP} for the case of mixed states in a Hilbert space of arbitrary dimension. To the best of our knowledge, this result was only formally proved for quantum states belonging to a Hilbert space of dimension~2~\cite{deVicenteThesis}. The extension of LPI for mixed states in arbitrary dimensions is one of the most important results of the present work.

\subsection{A family of uncertainty relations based on fidelity}
\label{ssec:familyURs}

We now show that a family of URs can be established in terms of \emph{fidelity-based metrics}. The procedure to obtain these URs is inspired in the one followed to obtain inequality~\eqref{eq:LPImix7}.

Let us start from fidelity-based metrics of the form
\begin{equation}
d(\rho,\sigma)= f\left(F(\rho,\sigma)\right)
\label{eq:fidelitymetric}
\end{equation}
where $f(x)$ is a {\em decreasing} function for $x\in[0,1]$, with $f(x)=0 \ \mathrm{iff} \ x=1$. Now, recalling the link between fidelity and probability given in Eq. ~\eqref{eq:Fprob}, we propose as uncertainty measure for an observable $\mathcal{A}$ the following quantity
\begin{equation}\label{eq:uncmmeasure}
\mathcal{U}(\mathcal{A};\rho) = f\left( P_{\mathcal{A};\rho} \right) .
\end{equation}
Notice that indeed this is a reasonable measure of uncertainty as it satisfies, by definition, the following basic properties:
\begin{enumerate}
\item $\mathcal{U}(\mathcal{A};\rho) \geq 0$,
\item $\mathcal{U}(\mathcal{A};\rho)$ is decreasing in terms of $P_{\mathcal{A};\rho}$, that is, uncertainty decreases when one has more certainty about the predictability of $\mathcal{A}$,
\item $\mathcal{U}(\mathcal{A};\rho) = 0$ if and only if $P_{\mathcal{A};\rho}=1$, that is, uncertainty vanishes only when one has certainty about the predictability of $\mathcal{A}$, and
\item the maximum of $\mathcal{U}(\mathcal{A};\rho)$ is attained at $P_{\mathcal{A};\rho}=\frac{1}{N}$, which leads to the uniform distribution, that is, uncertainty is maximum only when one has complete ignorance about the predictability of $\mathcal{A}$.
\end{enumerate}

We now give our main result in the following proposition that establishes a geometric formulation of uncertainty principle:
\begin{proposition} 
Let $\mathcal{A}$ and $\mathcal{B}$ be two observables with discrete non-degenerate spectra acting on an $N$-dimensional Hilbert space. Consider a quantum system described by a density operator $\rho$, and an uncertainty measure associated with the observables given in the form of Eq.~\eqref{eq:uncmmeasure}. Then, the following UR holds:
\begin{equation}\label{eq:URgeneral}
\mathcal{U}(\mathcal{A};\rho) + \mathcal{U}( {\mathcal{B};\rho} ) \geq f( c^2 )
\end{equation}
where $c^2 = \max_{i,j} \Tr (\Pi^{\mathcal{A}}_i \Pi^{\mathcal{B}}_j)$, being $\Pi^{\mathcal{A}}_i$ and $\Pi^{\mathcal{B}}_j$ the rank-one projectors associated with the $i$th outcome of $\mathcal{A}$ and the $j$th outcome of $\mathcal{B}$, respectively.

\label{prop:Urgeneral}
\end{proposition}

\begin{proof}
The triangle inequality fulfilled by~\eqref{eq:fidelitymetric} applied to the triplet $\Pi^{\mathcal{A}}_i$, $\Pi^{\mathcal{B}}_j$, $\rho$, leads to
\begin{equation} \label{eq:triangle}
 f\left( p_i(\mathcal{A};\rho) \right) + f\left( p_j(\mathcal{B};\rho) \right) \geq
 f\left( \Tr (\Pi^{\mathcal{A}}_i \Pi^{\mathcal{B}}_j) \right)
\nonumber
\end{equation}
where we made use of~\eqref{eq:Fprob}.
Note that this inequality is valid for \emph{any} pair of indices $i,j$. In particular, one can choose (separately) the indices $i_{\max}$ and $j_{\max}$ that correspond to the maximum probabilities $P_{\mathcal{A};\rho}$ and $P_{\mathcal{B};\rho}$, respectively. Using~\eqref{eq:uncmmeasure} we arrive to
\begin{equation} \label{eq:triangle2}
 \mathcal{U}(\mathcal{A};\rho) + \mathcal{U}(\mathcal{B};\rho) \geq
 f\left( \Tr (\Pi^{\mathcal{A}}_{i_{\max}} \Pi^{\mathcal{B}}_{j_{\max}}) \right),
\nonumber
\end{equation}
The proof concludes by using the fact that $\Tr (\Pi^{\mathcal{A}}_{i_{\max}} \Pi^{\mathcal{B}}_{j_{\max}}) \leq \max_{i,j} \Tr (\Pi^{\mathcal{A}}_i \Pi^{\mathcal{B}}_j)$ and that $f$ is decreasing.

\end{proof}

We remark that, regardless of the explicit form of the uncertainty measure $\mathcal{U}$, our formulation captures the essence of the uncertainty principle, in the sense discussed in Refs.~\cite{Deutsch1983,UffinkThesis}, due to the following reasons:
\begin{itemize}
\item[-] the lower bound to the uncertainty-sum is universal, that is, it is state-independent,
\item[-] when $c < 1$ the UR given by Eq.~\eqref{eq:URgeneral} is non-trivial, that is, the uncertainty-sum of the observables is strictly greater than zero, and
\item[-] when $c=\frac{1}{\sqrt N}$ (complementary observables), certainty associated with one observable implies maximum ignorance about the other. 
\end{itemize}
Furthermore, Eq.~\eqref{eq:URgeneral} represents in fact a family of URs.

\subsection{Comparisons among some uncertainty relations}
\label{ssec:comparison}

To conclude this work, we particularize the UR given by Eq.~\eqref{eq:URgeneral} for the case of the fidelity-based metrics introduced in Sec.~\ref{ssec:f-metrics}, and we compare the concomitant URs. In order to facilitate the comparison, we include the results previously obtained for the angle metric in addition to the new results obtained for Bures and root-infidelity metrics.

\begin{itemize}

\item {Angle metric:}\label{sssec:Ametric}

From Eqs.~\eqref{eq:anglemetric}, \eqref{eq:fidelitymetric} and \eqref{eq:uncmmeasure}, the
corresponding uncertainty measure is \ $\mathcal{U}_{\rm A}(\mathcal{O};\rho )= \arccos \sqrt{P_{\mathcal{O};\rho}}$ \ for $\mathcal{O}=\mathcal{A},\mathcal{B}$. Therefore, we reobtain inequality~\eqref{eq:LPImix7}:
\begin{equation}
 \arccos \sqrt{P_{\mathcal{A};\rho}} + \arccos \sqrt{P_{\mathcal{B};\rho}} \geq \arccos c.
 \label{eq:LPUR}
\end{equation}
Again we observe that this inequality is the natural extension of LPI~\eqref{eq:LP} to the case of mixed states belonging to a Hilbert space of {\em arbitrary dimension}. To the best of our knowledge, this result has been proved {\em only} for 2-dimensional states~\cite{deVicenteThesis}, where the fact that the uncertainty measure $\mathcal{U}_{\rm A}(\mathcal{O};\rho )$ is concave in terms of $\rho$ is crucial for the demonstration. However, that argument is not applicable to the case $N>2$ since the uncertainty measure looses the concavity property.

\item{Bures metric:}\label{sssec:Bmetric}

{}From Eqs.~\eqref{eq:buresmetric}, \eqref{eq:fidelitymetric} and \eqref{eq:uncmmeasure}, the corresponding uncertainty measure is \ $\mathcal{U}_{\rm B}(\mathcal{O};\rho )= \sqrt{2 - 2\sqrt{ P_{\mathcal{O};\rho}}}$ \ for $\mathcal{O}=\mathcal{A},\mathcal{B}$. Therefore, the following UR holds
\begin{equation}
\sqrt{1 - \sqrt{P_{\mathcal{A};\rho}}} + \sqrt{1 - \sqrt{P_{\mathcal{B};\rho}}} \geq \sqrt{1 - c}.
\label{eq:BuresUR}
\end{equation}

\item{Root-infidelity metric:}\label{sssec:RImetric}

{}From Eqs.~\eqref{eq:cmetric}, \eqref{eq:fidelitymetric} and \eqref{eq:uncmmeasure}, the corresponding uncertainty measure is \ $\mathcal{U}_{\rm RI}(\mathcal{O};\rho)= \sqrt{1 - P_{\mathcal{O};\rho}}$ \ for $\mathcal{O}=\mathcal{A},\mathcal{B}$. Therefore, the following UR holds
\begin{equation}
\sqrt{ 1- P_{\mathcal{A};\rho}} + \sqrt{ 1- P_{\mathcal{B};\rho}} \geq \sqrt{ 1- c^2}. \label{eq:CUR}
\end{equation}

\end{itemize}

\hfill

Now, let us compare the three URs~\eqref{eq:LPUR}, \eqref{eq:BuresUR} and~\eqref{eq:CUR} by focusing on the set of values of $P_{\mathcal{A};\rho}$ and $P_{\mathcal{B};\rho}$ allowed in each case. With this purpose, we define
\begin{equation}
\label{eq:domain}
\mathcal{D}_{\lambda,c}= \left\{ \left(P_{\mathcal{A};\rho},P_{\mathcal{B};\rho} \right) \in \left[\frac{1}{N},1\right] \times \left[\frac{1}{N},1\right]: \ P_{\mathcal{B};\rho} \leq g_{\lambda,c}(P_{\mathcal{A};\rho}) \right\}
\end{equation}
where $\lambda= {\rm A}$ (for angle metric), ${\rm B}$ (for Bures metric), or ${\rm RI}$ (for root-infidelity metric), and
\begin{equation}
g_{\lambda,c}(P)=\left\{
\begin{array}{ll}
1 & \mathrm{if} \ \frac 1N \leq P \leq c^2 \\
h_{\lambda,c}(P) & \mathrm{if} \ c^2\leq P\leq 1
\end{array}
\right.
\label{eq:gdomain}
\end{equation}
with $h_{\lambda,c}$ given in Table~\ref{tab:table2} (see App.~\ref{appendix}).
\begin{table}[ht]
\begin{center}
\begin{tabular}{cc}
\hline \hline
 Metric & $h_{\lambda,c}(P)$
 \\
 \hline \hline
 $d_{\rm A}$ (Angle) & $\left( \sqrt{1- P} \sqrt{1-c^2} + c \sqrt{P} \right)^2$
 \\
 $d_{\rm B}$ (Bures) & $\left( \sqrt{P}+ 2\sqrt{1 - \sqrt{P}} \sqrt{1-c} + c - 1 \right)^2$
 \\
$d_{\rm RI}$ (Root-infidelity) & $P + 2\sqrt{1- P} \sqrt{1-c^2} +c^2-1$
 \\
 \hline
\end{tabular}
\end{center}
\caption{Functions $h_{\lambda,c}$.}
\label{tab:table2}
\end{table}

It can be seen that the following ordering among the sets holds
\begin{equation}
\mathcal{D}_{\textrm{A},c} \subseteq \mathcal{D}_{\textrm{B},c} \subseteq \mathcal{D}_{\textrm{RI},c}
\end{equation}
for every $c$.
This implies that inequality~\eqref{eq:LPImix7}, derived for the angle metric, is the tightest one. However, two limiting cases arise where the three sets are equal.
One case is that in which one has certainty about one observable: if $P_{\mathcal{A};\rho}=1$, then $P_{\mathcal{B};\rho} \leq g_{\lambda,c}(1) = c^2$ for any metric (and analogously interchanging $\mathcal{A}$ and $\mathcal{B}$). The other case is trivial: when $c=1$, the three sets are the whole square $\left[\frac{1}{N},1\right] \times \left[\frac{1}{N},1\right]$, that is, there is no restriction coming from URs.
To illustrate these results, we show in Fig.~\ref{f:fig1} the sets given by Eq.~\eqref{eq:domain} for typical values of the overlap~$c$.

\begin{figure}[htbp]
\begin{center}
\includegraphics[width=15cm]{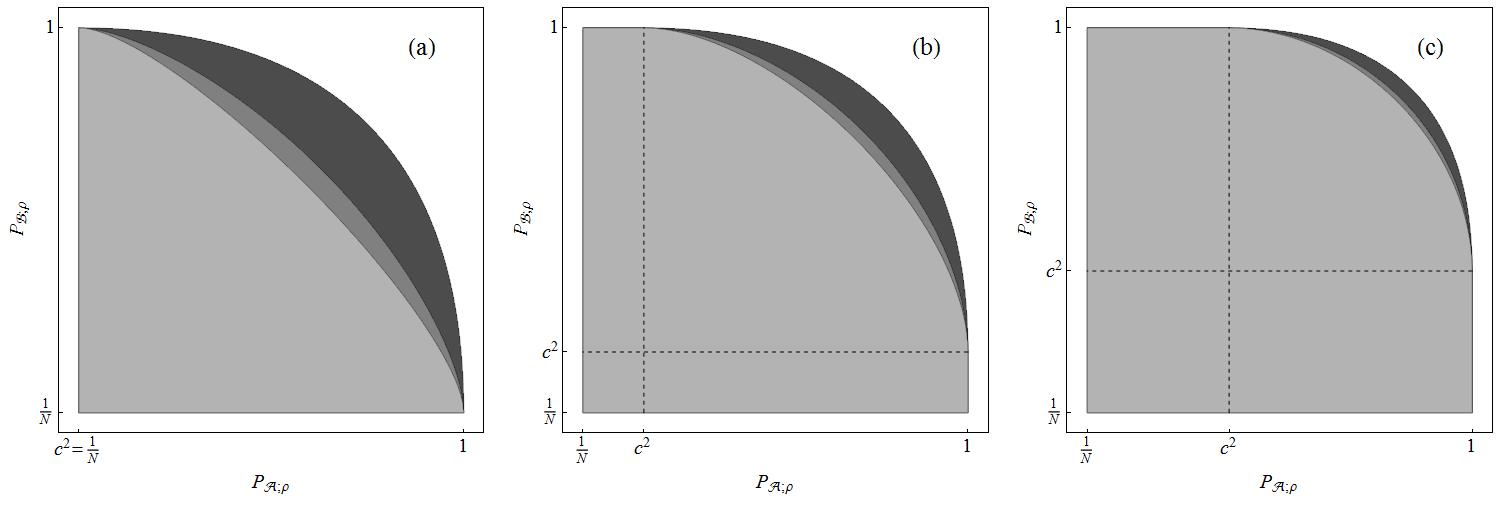}
\end{center}
\caption{Plots of the sets $\mathcal{D}_{\lambda,c}$ given in Eq.~\eqref{eq:domain} where $\lambda$ refers to angle metric (A, in light gray), Bures metric (B, in gray) and root-infidelity metric (RI, in dark gray), when $N=20$, for: (a)~$c=\frac{1}{\sqrt N}$, (b)~$c = \sqrt{0.2}$, and (c)~$c=\sqrt{0.4}$.}
\label{f:fig1}
\end{figure}

\section{Concluding remarks}
\label{sec:conclu}

In this work, we present a geometric approach to formulate the uncertainty principle. We obtain a family of uncertainty relations, Eq.~\eqref{eq:URgeneral}, that depend on fidelity-based metrics. In particular, when we make use of the angle metric between quantum sates, a natural generalization of the Landau--Pollak inequality to mixed states in arbitrary dimensions is obtained, i.e.\ Eq.~\eqref{eq:LPUR}.
In addition we find two novel uncertainty relations, Eqs.~\eqref{eq:BuresUR} and~\eqref{eq:CUR}, derived from Bures and root-infidelity metric, respectively. These relations are seen to be weaker than Landau--Pollak inequality. As a consequence of these findings, it rises up the question whether the angle metric leads to the tightest uncertainty relation  when compared to an uncertainty relation derived from any {\em arbitrary} fidelity-based metric, within our approach. This observation deserves  further study. In addition, it remains open the problem of how to extend our main result~\eqref{eq:URgeneral} to the most general case of Positive-Operator Valued Measures.

\acknowledgments

TMO, PWL and MP are members of the Consejo Nacional de Investigaciones Cient\'{\i}ficas y T\'ecnicas (CONICET), Argentina. GMB is a fellow of CONICET. GMB and MP acknowledge Secretar\'{\i}a de Pol\'{\i}ticas Universitarias for financial support through the mobility program Inter-U, and acknowledge warm hospitality at FaMAF, UNC. PWL and TMO are grateful to Secretar\'{\i}a de Ciencia y T\'ecnica, Universidad Nacional de C\'ordoba (SECyT-UNC) for financial support.


\appendix
\section{Set of allowed values of $P_{\mathcal{A},\rho}$ and $P_{\mathcal{B},\rho}$}
\label{appendix}

In this appendix we show how to obtain the sets~\eqref{eq:domain} of allowed values of $P_{\mathcal{A};\rho}$ and $P_{\mathcal{B};\rho}$ for each metric. After some algebra, URs~\eqref{eq:LPUR}, \eqref{eq:BuresUR} and~\eqref{eq:CUR} can be written as,
\begin{equation}
 \xi^2 + a_1 \xi + a_0 \geq 0
\end{equation}
where $\xi$, $a_1$ and $a_0$ are given in Table~\ref{tab:table1} for each metric.
\begin{table}[ht]
\begin{center}
\begin{tabular}{cccc}
\hline \hline
 Metric & $\xi$ & $a_1$ & $a_0$
 \\
 \hline \hline
$d_{\rm A}$ (Angle) & $\sqrt{1- P_{\mathcal{B};\rho}}$ & $2c \sqrt{1-P_{\mathcal{A};\rho}}$ & $c^2 - P_{\mathcal{A};\rho}$
 \\
$d_{\rm B}$ (Bures) & $\sqrt{2-2\sqrt{P_{\mathcal{B};\rho}}}$ & $2\sqrt{2-2\sqrt{P_{\mathcal{A};\rho}}}$ & $2(c-\sqrt{P_{\mathcal{B};\rho}})$
 \\
$d_{\rm RI}$ (Root-infidelity) & $\sqrt{1- P_{\mathcal{B};\rho}}$ & $2 \sqrt{1-P_{\mathcal{A};\rho}}$ & $c^2 - P_{\mathcal{A};\rho}$
\\
\hline
\end{tabular}
\end{center}
\caption{Values of the variable $\xi$ and the coefficients of the quadratic polynomial for each metric.}
\label{tab:table1}
\end{table}
The quadratic polynomial corresponding to each metric has two roots $\xi_\pm$ being $\xi_-$ always negative. Thus, two cases arise. If $P_{\mathcal{A};\rho} \leq c^2$ then $\xi_+ \leq 0$ and $P_{\mathcal{B};\rho} \leq 1$, otherwise $\xi_+ \geq 0$ and $P_{\mathcal{B};\rho} \leq h_c\left( P_{\mathcal{A};\rho} \right)$ with $h_{\lambda,c}$ given in Table~\ref{tab:table2}.



\end{document}